\documentclass[11pt,reqno]{amsart}

\usepackage[T1]{fontenc}
\usepackage{lmodern}
\usepackage{microtype}
\usepackage{amsmath,amssymb,amsthm,mathtools}
\usepackage{mathrsfs}
\usepackage{enumitem}
\usepackage{comment}
\usepackage[colorlinks=true,linkcolor=blue,citecolor=blue,urlcolor=blue]{hyperref}

\usepackage{mathtools}
\mathtoolsset{showonlyrefs}

\hypersetup{
  pdftitle={Electrovacuum black hole uniqueness via singular harmonic maps},
  pdfauthor={Qing Han, Marcus Khuri, Gilbert Weinstein, and Jingang Xiong},
  pdfsubject={Electrovacuum black hole uniqueness and logarithmic angle defects for singular harmonic maps},
  pdfkeywords={electrovacuum black hole uniqueness, Kerr--Newman, Majumdar--Papapetrou, singular harmonic maps, horizon rods, logarithmic angle defects, maximum principle}
}

\allowdisplaybreaks
\numberwithin{equation}{section}

\newtheorem{theorem}{Theorem}[section]
\newtheorem{proposition}[theorem]{Proposition}
\newtheorem{lemma}[theorem]{Lemma}
\newtheorem{corollary}[theorem]{Corollary}
\newtheorem{remark}[theorem]{Remark}
\newtheorem{definition}[theorem]{Definition}

\newcommand{\R}{\mathbb R}
\newcommand{\C}{\mathbb C}
\newcommand{\CH}{\mathbb H_{\C}^{2}}

\newcommand{\Th}{\Theta}
\newcommand{\grad}{\nabla}
\newcommand{\Dtwo}{\Delta_{\rho,z}}

\newcommand{\cE}{\mathcal E}
\newcommand{\cT}{\mathcal T}
\newcommand{\cA}{\mathcal A}
\newcommand{\starTwo}{\star_{2}}

\title[Electrovacuum Black Hole Uniqueness]{Electrovacuum Black Hole Uniqueness}

\author[Q. Han]{Qing Han}
\address{Department of Mathematics, University of Notre Dame, Notre Dame, IN 46556, USA}
\email{qhan@nd.edu}

\author[M. Khuri]{Marcus Khuri}
\address{Department of Mathematics, Stony Brook University, Stony Brook, NY 11794, USA}
\email{marcus.khuri@stonybrook.edu}

\author[G. Weinstein]{Gilbert Weinstein}
\address{Department of Mathematics and Department of Physics, Ariel University, Ariel, 40700, Israel}
\email{gilbertw@ariel.ac.il}

\author[J. Xiong]{Jingang Xiong}
\address{School of Mathematical Sciences, Laboratory of Mathematics and Complex Systems, MOE, Beijing Normal University, Beijing 100875, China}
\email{jx@bnu.edu.cn}

\thanks{Q. Han acknowledges the support of NSF Grant DMS-2305038. M. Khuri acknowledges the support of NSF Grant DMS-2405045. J. Xiong acknowledges the partial support of NSFC Grants 12325104. G. Weinstein acknowledges the support of ISF Grant 1119/26.}

\begin{document}

\begin{abstract}
We prove the black hole uniqueness conjecture in the axially symmetric, stationary, electrovacuum setting,
subject to the refined asymptotic analysis of the associated singular harmonic maps, which includes an analyticity hypothesis at the axes. More precisely, it is shown
that any asymptotically flat solution of the Einstein--Maxwell equations in this class, with more than one black hole
horizon component is either: Majumdar--Papapetrou, up to a duality rotation, in which case all logarithmic angle
defects vanish, or every finite axis rod logarithmic angle defect is strictly negative and hence every
interaction force is strictly attractive.
The proof extends the singular harmonic map method used for vacuum Kerr
uniqueness in \cite{HKWXvacuum}.  
\end{abstract}

\maketitle

\section{Introduction}
\label{sec1}

A fundamental question of general relativity asks to what extent a stationary black hole
is determined by its conserved charges.  For the
four-dimensional Einstein--Maxwell equations, the expected connected-horizon
answer is the Kerr--Newman family.  The classical proofs are based on the
Ernst formulation \cite{Ernst1968}, divergence identities, and the harmonic map
structure of the reduced field equations; see Carter
\cite{Carter1971,Carter1973,Carter1985},
Robinson \cite{Robinson1975}, and Mazur \cite{Mazur1982}.
More recently in \cite{Costa2010}, Costa
classified stationary, $I^+$-regular, analytic, nondegenerate electrovacuum
black holes in terms of Weinstein solutions and obtained Kerr--Newman when the
horizon is connected, see also the discussion in \cite[Section 3]{ChruscielCostaHeusler2012}.  The connected rotating degenerate case
was treated by Amsel--Horowitz--Marolf--Roberts
\cite{AmselEtAl2010} and by Chru\'sciel--Nguyen
\cite{ChruscielNguyen2010}. Moreover, a constructive axis-potential proof, including the
extreme case, was given by Meinel \cite{Meinel2012}.

Disconnected horizons introduce different phenomena.  Recall that in vacuum,
there are no regular static multi-black hole solutions
\cite{BuntingMasood1987,ChruscielReallTodStatic2006}, and the conical
singularities in the Bach--Weyl and Israel--Khan metrics
\cite{BachWeyl1922,IsraelKhan1964} represent the struts needed to prevent
collapse. On the other hand, electromagnetic repulsion creates the possibility of
regular balanced solutions. Indeed, Majumdar \cite{Majumdar1947} and Papapetrou
\cite{Papapetrou1947} found the explicit static solutions
\begin{equation}\label{mp000}
 \mathbf{g}_{mp}=-\mathcal{U}_{mp}^{-2}dt^2+\mathcal{U}_{mp}^2\,d\mathbf{x}^2,\qquad
 \mathcal{U}_{mp}=1+\sum_{i=1}^{N}\frac{m_i}{|\mathbf{x}-p_i|},\qquad m_i>0,
\end{equation}
on $\mathbb{R}\times(\mathbb{R}^3 \setminus \cup_{i=1}^{N} \{p_i\})$,
and Hartle--Hawking \cite{HartleHawking1972} showed that the apparent point
singularities are in fact degenerate horizons. Here, $m_i$ represents the mass and charge of each black hole.  
Chru\'sciel--Nadirashvili \cite{ChruscielNadirashvili1995} proved that a static, asymptotically flat electrovacuum black hole locally of Majumdar--Papapetrou type, with nonsingular domain of outer communications, is necessarily a standard Majumdar--Papapetrou solution. Furthermore, the static
electrovacuum classification of Chru\'sciel--Tod \cite{ChruscielTod2007} identifies the alternatives of
Reissner--Nordstr\"om and Majumdar--Papapetrou under certain global hypotheses. We mention also the analysis of Israel--Wilson--Perj\'es black holes in \cite{ChruscielReallTod2006}, the Majumdar--Papapetrou classification of Lucietti \cite{Lucietti}, and the recent related work of Hirsch--Zhang \cite{HirschZhang}.

The remaining balance problem is to understand rotating and charged
configurations with several horizon components.  In Weyl--Papapetrou
coordinates, the stationary axisymmetric Einstein--Maxwell equations reduce
to a singular harmonic map into complex hyperbolic plane.  Weinstein used this
structure to construct all possible configurations of nondegenerate horizons
\cite{Weinstein1995,Weinstein1996}, while Khuri--Weinstein \cite{KhuriWeinstein2016} treated the case of degenerate horizons. 
The reconstructed metrics are smooth away
from possible conical singularities on finite axis rods.  Regularity of the reduced harmonic maps at axis interiors was established by Nguyen \cite{Nguyen2011}. An alternative approach, via inverse scattering techniques, gives a complementary finite dimensional
description of the charged balance conditions, see Hennig
\cite{Hennig2019,Hennig2020,Hennig2026}.
Nevertheless, the question of determining which
charged multi-black hole configurations can exist in equilibrium without struts, has stayed largely open until now.

In the vacuum setting, the authors recently proved that every finite axis rod
has negative logarithmic angle defect, and hence positive attractive force,
for arbitrary finite mixed configurations of nondegenerate and rotating
degenerate horizons \cite{HKWXvacuum}.  The proof combines refined puncture
asymptotics from \cite{HKWXasymptotics} with a global differential inequality
for the Weyl conformal factor, and a local analysis at special points along axis
rods.  The purpose of the present work is to generalize this result to the electrovacuum
setting, and identify the Majumdar--Papapetrou family of solutions as the only balanced configurations.

We now describe the setting of the main theorem. Let $\Gamma=\{\rho=0\}\subset\R^3$
denote the rotation $z$-axis in cylindrical coordinates
$(\rho,z,\phi)$, with $\phi$ of period $2\pi$.  On the domain of outer
communications away from the orbit space boundary, the spacetimes under consideration have
topology $\mathbb{R}\times (\mathbb{R}^3 \setminus\Gamma)$, and their axisymmetric stationary
electrovacuum metrics take the Weyl--Papapetrou form
\begin{equation}\label{eq:WP}
 \mathbf g=-e^{2U}dt^2+\rho^2e^{-2U}(d\phi+w\,dt)^2
 +e^{-2U+2\alpha}(d\rho^2+dz^2).
\end{equation}
Here $\partial_t$ and $\partial_\phi$ are respectively the
stationary and rotational Killing fields. 
Setting $u=U-\ln\rho$, Einstein--Maxwell equations reduce to solving for an axisymmetric harmonic map
\begin{equation}\label{lo}
 \Phi=(u,v,\chi,\psi):\mathbb R^3\setminus\Gamma\longrightarrow\mathbb H^2_{\mathbb C},
\end{equation}
where the target complex hyperbolic plane has metric
\begin{equation}\label{eq:CHmetric}
 h_{\mathbb H^2_{\mathbb C}}
 =du^2+e^{4u}(dv-\psi\,d\chi+\chi\,d\psi)^2
   +e^{2u}(d\chi^2+d\psi^2).
\end{equation}
The functions $v,\chi,\psi$ represent the charged twist and electromagnetic
potentials.  The full equations, and reconstruction quadratures which derive $\alpha$ and $w$ from the harmonic map, are recorded in Section~\ref{sec:singularmap}.

The boundary $\{\rho=0\}$ of the orbit half-plane decomposes into
alternating horizon and axis intervals.  Write
\[
 \mathcal H_i=[z_i^-,z_i^+],\qquad
 z_i^-\leq z_i^+<z_{i+1}^-,\qquad 1\leq i\leq N,
\]
for the horizon components.  If $z_i^-<z_i^+$, the interior of
$\mathcal H_i$ is a nondegenerate horizon rod, and the endpoints $(0,0, z_i^-)$ and $(0,0,z_i^+)$ are the
lower and upper poles of the horizon.  If
$z_i^-=z_i^+=z_i$, the horizon is degenerate and is represented in the
orbit space by the puncture $p_i=(0,0,z_i)\in\R^3$.
The complementary axis rods are
\[
 \Gamma_1=(-\infty,z_1^-),\qquad
 \Gamma_{j+1}=(z_j^+,z_{j+1}^-),\qquad
 \Gamma_{N+1}=(z_N^+,\infty),\qquad 1\leq j\leq N-1.
\]
On each $\Gamma_j$, the three potentials have constant traces.  With a fixed
orientation and electromagnetic gauge, their jumps encode the individual
angular momenta and charges of the horizons $\mathcal{H}_i$, by
\begin{equation}
 \mathcal{J}_i=\frac14(v_{i+1}-v_i),\qquad
 Q_i^e=\frac12(\chi_{i+1}-\chi_i),\qquad
 Q_i^b=\frac12(\psi_{i+1}-\psi_i),\qquad 1\leq i\leq N,
\end{equation}
as in \cite{KhuriWeinstein2016}. The total charge will be denoted by $Q_i=\sqrt{(Q_i^e)^2+(Q_i^b)^2}$.

The \textit{logarithmic angle defects} $\mathbf{b}_j$ also arise directly from the
reconstructed harmonic map data. More precisely, standard axis regularity implies that $\alpha$ has a constant
trace on each axis rod.  For $z_0\in\Gamma_j$, the
appropriate radius/circumference ratio computed from
\eqref{eq:WP} yields
\begin{equation*}
 e^{\mathbf b_j}
 :=\lim_{\rho\rightarrow 0}
 \frac{\displaystyle\int_0^\rho
 e^{\alpha(s,z_0)-U(s,z_0)}\,d s}
 {\rho e^{-U(\rho,z_0)}}
 =e^{\alpha(0,z_0)} \qquad\Rightarrow\quad\quad
 \mathbf b_j=\alpha|_{\Gamma_j}.
\end{equation*}
The corresponding interaction force \cite{Weinstein1990} may then be expressed as
\begin{equation*}
 \mathcal F_j=\frac14\left(e^{-\mathbf b_j}-1\right),
\end{equation*}
so that $\mathbf b_j<0$ is equivalent to a strictly positive, or attractive, force on the axis rod. In the asymptotically flat setting, $\alpha\to0$ at infinity, and hence the two
semi-infinite rods satisfy
\begin{equation}\label{eq:outerDefects}
 \mathbf b_1=\mathbf b_{N+1}=0.
\end{equation}

Section~\ref{sec:singularmap} gives the definition of a \textit{stationary
electrovacuum singular harmonic map} with black hole rod data, which is used in the statement of our main theorem. This notion records the standard horizon, axis, and infinity asymptotics of the harmonic maps that arise from asymptotically flat,
axisymmetric, stationary electrovacuum spacetimes.

\begin{theorem}\label{thm:main}
Let $N>1$, and let
\[
 \Phi=(u,v,\chi,\psi):\mathbb R^3\setminus\Gamma\longrightarrow
 \mathbb H^2_{\mathbb C}
\]
be a stationary electrovacuum singular harmonic map with $N$ black hole
horizon components.  Assume that
\[
 \mathcal{J}_i \neq 0,\quad\quad\text{ or }\qquad Q_i \neq 0,
\]
for every degenerate horizon component; no such hypothesis is imposed on
nondegenerate horizon rods. 
Then exactly one of the following alternatives holds.
\begin{enumerate}[label=\textup{(\roman*)},leftmargin=2.5em]
\item\label{alt:attractive}
Every finite axis rod has negative logarithmic angle defect,
\begin{equation}
 \mathbf{b}_j<0,\qquad 2\leq j\leq N.
\end{equation}

\item\label{alt:MP}
Every horizon component is degenerate, all axis defects vanish, and $\Phi$ is the harmonic map
of a Majumdar--Papapetrou spacetime, up to a Heisenberg translation in the target space and a constant duality rotation of the Maxwell field.
\end{enumerate}
\end{theorem}

\begin{remark}
The conditions that characterize stationary electrovacuum singular harmonic maps, in Definition \ref{def:blackHoleData}, including analyticity at the axes, arise as consequences of substantially weaker assumptions \cite{HKWXprog}, see the discussion in Section \ref{sec2.2}.
\end{remark}

The theorem is formulated entirely at the level of singular harmonic
maps.  Its standard spacetime consequence may be briefly summarized as follows.
Suppose that an analytic, asymptotically flat, stationary electrovacuum
black hole spacetime satisfies the usual $I^+$-regularity hypotheses.  If at least one component of the event horizon is
rotating, the rigidity theorem \cite{ChruscielRigidity}, yields a global axial Killing field,
independently of horizon connectedness or degeneracy. The electrovacuum
circularity theorem \cite{ChruscielCostaHeusler2012} then gives the Weyl--Papapetrou reduction, and corresponding
stationary electrovacuum singular harmonic map.  If the horizon were disconnected,
regularity of the finite axis rods would rule out
alternative~\ref{alt:attractive} of Theorem~\ref{thm:main}, while the
remaining alternative would imply that the spacetime is
Majumdar--Papapetrou and hence static, contradicting the presence of a
rotating horizon component.  The horizon is therefore connected, and the
known connected-horizon classifications imply that the domain of
outer communications is Kerr--Newman: Costa \cite{Costa2010} in the nondegenerate
case, and Chru\'sciel--Nguyen \cite{ChruscielNguyen2010} or
Amsel--Horowitz--Marolf--Roberts \cite{AmselEtAl2010} in the degenerate case.

If, on the other hand, every horizon component is nonrotating, axial
symmetry does not follow merely from analyticity, as the Majumdar--Papapetrou family demonstrates.  
However, if the domain of outer communications admits a smooth maximal Cauchy surface,
then the electrovacuum staticity theorem \cite{SudarskyWald} implies that the solution is
static, and the classification of Chru\'sciel--Tod
\cite{ChruscielTod2007} gives Reissner--Nordstr\"om when the horizon is
connected and Majumdar--Papapetrou when it is disconnected.  Thus, under
the rigidity theorem hypotheses in the rotating case and the 
staticity theorem hypotheses in the nonrotating case, the
Kerr--Newman/Majumdar--Papapetrou dichotomy follows. 
In this overall picture, Theorem~\ref{thm:main} provides the key step in the rotating case by excluding disconnected 
event horizons.

We now outline the proof for the main theorem. A central observation is a new nonlinear differential inequality for the
Weyl conformal factor $\alpha$. More precisely, by differentiating the Weyl quadratures equations that determine $\alpha$, and exploiting the complex-square structure of the harmonic map stress-energy tensor we arrive at
\begin{equation}\label{eq:introAlphaIneq}
 \Delta_{\rho,z}\alpha+\rho^{-1}|\nabla\alpha|
 \geq 2e^{4u}|dv-\psi\,d\chi+\chi\,d\psi|^2\geq0.
\end{equation}
By invoking the maximum principle on a sequence of exhaustion domains that are separated from the axis in the orbit space half-plane, and using the asymptotics at nondegenerate horizon poles and degenerate horizon punctures, we obtain the global upper bound
\begin{equation}\label{ur}
 \alpha\leq\max_{1\leq j\leq N+1}\mathbf{b}_j .
\end{equation}
This conclusion is analogous to the first goal in the proof for the vacuum setting \cite{HKWXvacuum}. At this point, the strategies for the vacuum and electrovacuum cases diverge, since here we intend to take advantage of the analyticity condition
of the reduced harmonic map along axis rods, while in \cite{HKWXvacuum} such analyticity is not assumed.
Let $\Gamma_j$ be a finite axis rod which realizes the largest logarithmic angle defect. Near this rod analyticity 
yields the factorizations
\[
 U=F(\rho^2,z),\qquad v=\rho^4G(\rho^2,z), \qquad \chi=\rho^2H(\rho^2,z),\qquad
 \psi=\rho^2K(\rho^2,z),
\]
and at an interior maximum for $U|_{\Gamma_j}$ we show that
\[
 \alpha_{\rho\rho}(0,z)
 =4e^{2F(0,z)}\bigl(H(0,z)^2+K(0,z)^2\bigr).
\]
The global upper bound forces this expression to vanish at the maximum. This, combined with the differential
inequality \eqref{eq:introAlphaIneq}, and further asymptotic analysis, gives $\alpha\equiv0$. 
The conclusions of Theorem \ref{thm:main} then follow in a relatively straightforward manner.

The paper is organized as follows.  Section~\ref{sec:singularmap} records the
complex-hyperbolic reduction and the harmonic map asymptotic class.  Section \ref{sec3}
establishes the Weyl conformal factor differential inequality \eqref{eq:introAlphaIneq}, and
the global upper bound \eqref{ur}.
The axis factorization of $\alpha$, and a related quadratic boundary rigidity theorem, are
provided in Section \ref{sec4}. 
The proof of the main theorem is given in Section \ref{sec:equality}.

\subsection*{Acknowledgements}
This project was made possible in part by a SQuaRE at the American Institute of Mathematics. The authors thank AIM for providing a supportive and mathematically rich environment. 

\section{The Electrovacuum Singular Harmonic Map}\label{sec:singularmap}

\subsection{The complex hyperbolic reduction}
By setting $u=U-\ln\rho$, the Einstein--Maxwell equations for the Weyl-Papapetrou metric \eqref{eq:WP}
reduce (\cite{Weinstein1996}) to an axisymmetric harmonic map system for $\Phi$. If
\begin{equation}
 \Th=d v-\psi\,d\chi+\chi\,d\psi,
\end{equation}
then, with all norms and differential operators taken with respect to the Euclidean metric on
$\R^3\setminus\Gamma$, the harmonic map equations may be expressed as
\begin{align}
\begin{split}
 \Delta u&=2e^{4u}|\Th|^2+e^{2u}
   \bigl(|\grad\chi|^2+|\grad\psi|^2\bigr),
   \label{eq:harmu}\\
 \operatorname{div}(e^{4u}\Th)&=0,
   \\
 \operatorname{div}(e^{2u}\grad\chi)
   &-2e^{4u}\Th\cdot\grad\psi=0,\\
 \operatorname{div}(e^{2u}\grad\psi)
   &+2e^{4u}\Th\cdot\grad\chi=0.
\end{split}
\end{align}
The remaining metric functions $w$ and $\alpha$ are recovered by quadrature, namely
\begin{align}
 w_\rho&=2\rho e^{4u}\Th_z,\qquad\quad\quad
 w_z=-2\rho e^{4u}\Th_\rho,
 \label{eq:wquad}\\
 \alpha_\rho
 &=\rho\left[U_\rho^2-U_z^2
 +e^{4u}(\Th_\rho^2-\Th_z^2)
 +e^{2u}(\chi_\rho^2-\chi_z^2+\psi_\rho^2-\psi_z^2)\right],
 \label{eq:alphaquad1}\\
 \alpha_z
 &=2\rho\left[U_\rho U_z+e^{4u}\Th_\rho\Th_z
 +e^{2u}(\chi_\rho\chi_z+\psi_\rho\psi_z)\right].
 \label{eq:alphaquad2}
\end{align}
Moreover, the spacetime field strength tensor for the electromagnetic field is related to
the electromagnetic potentials by
\begin{equation}\label{axialcon}
d\chi=\iota_{\partial_{\phi}} {*}_{\mathbf g}\mathbf F,\qquad
 d\psi=\iota_{\partial_{\phi}}\mathbf F,
\end{equation}
and given explicitly as
\begin{equation}\label{recem}
 \mathbf F
 =(d\phi+w\,d t)\wedge d\psi
 -{*}_{\mathbf g}
   \bigl((d\phi+w\,d t)\wedge d\chi\bigr),
\end{equation}
where ${*}_{\mathbf g}$ denotes the spacetime Hodge star and $\iota_{\partial_{\phi}}$ is interior product.

The integrability conditions for \eqref{eq:wquad}--\eqref{eq:alphaquad2} follow from
\eqref{eq:harmu}.  
Furthermore, since the target space metric and the harmonic map equations are real analytic, 
elliptic regularity shows that $\Phi$ is real analytic on $\{\rho>0\}$; the quadratures then imply that
$w$ and $\alpha$ are real analytic there as well. We also note that the complex hyperbolic target space $\mathbb{H}_{\mathbb{C}}^2$ is invariant under the Heisenberg translation isometries
\begin{equation}\label{eq:Heisenberg}
 v\longmapsto v+a_3\chi-a_2\psi+a_1,
 \qquad
 \chi\longmapsto\chi+a_2,
 \qquad
 \psi\longmapsto\psi+a_3,
\end{equation}
where $a_1,a_2,a_3$ are constants. Therefore, on any fixed axis rod, where $v,\chi,\psi$ have
constant traces, one may choose the constants in \eqref{eq:Heisenberg} so
that all three traces vanish, while the 1-form $\Th$ is unchanged.

\subsection{Regularity and asymptotics}\label{sec2.2}
We now state the local and asymptotic hypotheses for the harmonic maps appearing in the main theorem. The conditions below are modeled on the vacuum black hole class of maps in \cite{HKWXvacuum} and arise from: the axis and pole estimates of Nguyen \cite{Nguyen2011}, Li--Tian \cite{LiTian1992,LiTian1993}, Weinstein \cite{Weinstein1990,Weinstein1992,Weinstein1995}, as well as the puncture and infinity estimates of Han--Khuri--Weinstein--Xiong \cite{HKWXasymptotics}. The conditions concerning the Weyl conformal factor $\alpha$ follow directly from the harmonic map asymptotics and equations \eqref{eq:alphaquad1}--\eqref{eq:alphaquad2}. The primary difference between the prior works listed above, and the conditions appearing below, is the analyticity at the interior of axis rods. This, however, turns out not to be an additional restriction, see Remark \ref{rem:futureVerification}.

\begin{definition}
\label{def:blackHoleData}
Let $\Phi=(u,v,\chi,\psi)$ be an axisymmetric weak solution of
\eqref{eq:harmu} on $\mathbb R^3\setminus\Gamma$, having the prescribed potential constants and
horizon/axis rod structure along $\Gamma$ as in Section \ref{sec1}.  We say that $\Phi$ is a \emph{stationary electrovacuum singular harmonic map} with black hole rod data if the
following conditions hold.
\begin{enumerate}[label=\textup{(A\arabic*)},leftmargin=2.7em]
\item\label{A1}
\emph{Axes.} Let $\Gamma_j$ be an axis rod. For every compact subinterval $I\subset\Gamma_j$, after the
Heisenberg translation \eqref{eq:Heisenberg} which makes the traces of $v,\chi,\psi$ vanish along this rod,
there is a neighborhood $\mathcal O_I \subset\mathbb{R}^3$ and a constant $\Lambda_I$ such that
\begin{equation}
 |U|+e^u(|\chi|+|\psi|)\leq\Lambda_I
 \qquad\text{on }\mathcal O_I\setminus\Gamma_j,
\end{equation}
and the renormalized energy is finite
\begin{equation}
 \int_{\mathcal O_I}\left[
 |\nabla U|^2+e^{4u}|\Theta|^2
 +e^{2u}(|\nabla\chi|^2+|\nabla\psi|^2)
 \right]dx<\infty.
\end{equation}
Moreover, there exist real analytic functions $F,H,K,G$ on
$\mathcal O_I$ such that, with $s=\rho^2$,
\begin{equation}\label{eq:axisAnalyticAssumption}
 U=F(s,z),\qquad
 v=s^2G(s,z),\qquad
 \chi=sH(s,z),\qquad
 \psi=sK(s,z).
\end{equation}

\item\label{A2}
\emph{Nondegenerate horizons and poles.}  On the interior of a
nondegenerate horizon rod, $u,v,\chi,\psi$ are smooth, and in a neighborhood of each compact subinterval
\begin{equation}\label{eq:horizonexpansion}
 u(\rho,z)=u_0(z)+O_2(\rho^2),\qquad
 \alpha(\rho,z)=\ln\rho+O_2(1),
\end{equation}
for some smooth function $u_0$. At the lower and upper poles
$q_i^-=(0,0,z_i^-)$ and $q_i^+=(0,0,z_i^+)$, in polar coordinates
$(r_i^\pm,\theta_i^\pm)$ centered at the pole,
\begin{align}
 \alpha(r_i^-,\theta_i^-)
 &=\mathbf{b}_i+\ln\sin\frac{\theta_i^-}{2}
   +O((r_i^-)^{\gamma_i^-}),\label{eq:polealphaLower}\\
 \alpha(r_i^+,\theta_i^+)
 &=\mathbf{b}_{i+1}+\ln\cos\frac{\theta_i^+}{2}
   +O((r_i^+)^{\gamma_i^+}),\label{eq:polealphaUpper}
\end{align}
for some $\gamma_i^\pm>0$, uniformly in the angular variable.  On the
adjacent axes,
\begin{equation}\label{eq:poleU}
 U(0,z)=\frac12\ln|z-z_i^\pm|+O(1).
\end{equation}

\item\label{A3}
\emph{Degenerate horizons.}  
At each puncture $p_i=(0,0,z_i)$ either
\[
 \mathcal{J}_i \neq 0,\quad\quad\text{ or }\qquad Q_i \neq 0,
\]
and there is an axisymmetric
tangent harmonic map $\bar\Phi_i=\bar\Phi_i(\theta_i)$ from
$\mathbb{S}^2 \setminus\{\mathcal{N},\mathcal{S}\}\rightarrow \mathbb H^2_{\mathbb C}$, such that for any $\gamma\in(0,1)$,
\begin{equation}
(\bar U_i,\bar v_i,\bar\chi_i,\bar\psi_i)\in C^{3,\gamma}(\mathbb{S}^2),\qquad\quad  \bar U_i(\theta_i)=\bar u_i(\theta_i)+\ln\sin\theta_i,
\end{equation}
and
\begin{equation}\label{eq:tangentUConvergence}
 \sup_{0\leq\theta_i\leq\pi}
 \left|U(r_i,\theta_i)-\ln r_i-\bar U_i(\theta_i)\right|
 \longrightarrow0\qquad\text{as }r_i\rightarrow0,
\end{equation}
with corresponding estimates for derivatives up to order three, where $\mathcal{N},\mathcal{S}$ are the north and south poles. Moreover, there is a function
$\bar\alpha_i\in C^2([0,\pi])$ with
\begin{equation}\label{eq:tangentEndpoints}
 \bar\alpha_i(0)=\mathbf{b}_{i+1},\qquad
 \bar\alpha_i(\pi)=\mathbf{b}_i,
\end{equation}
such that the Weyl conformal factor satisfies
\begin{equation}
 \sup_{0\leq\theta_i\leq\pi}
 |\alpha(r_i,\theta_i)-\bar\alpha_i(\theta_i)|
 \longrightarrow0\qquad\text{as }r_i\rightarrow0,
\end{equation}
along with corresponding estimates for derivatives up to order two.

\item\label{A4}
\emph{Infinity.}  As $r\to\infty$,
\begin{equation}\label{eq:infinity}
 U=O(r^{-1}),\qquad \alpha=O(r^{-2}),\qquad
 \nabla\alpha=O(r^{-3}).
\end{equation}
In particular, the additive constant in $\alpha$ is fixed by $\alpha\rightarrow 0$; the same is done for the metric coefficient $w$.
\end{enumerate}
\end{definition}

\subsection{Existence and uniqueness} 
Consider a general \textit{rod data set}, consisting of a collection of horizon rods $\{\mathcal{H}_i \}_{i=1}^{N}$, axis rods $\{\Gamma_j \}_{j=1}^{N+1}$, and corresponding collections of potential constants $\{v_j\}_{j=1}^{N+1}$, $\{\chi_j\}_{j=1}^{N+1}$, $\{\psi_j\}_{j=1}^{N+1}$. By superposition of Kerr--Newman and extreme Kerr--Newman harmonic maps, one may construct an axisymmetric model map $\Phi_0:\mathbb{R}^3 \setminus\Gamma\rightarrow\mathbb{H}^2_{\mathbb{C}}$ which realizes this rod data set, and has the additional properties that its tension is globally bounded and decays quickly at infinity; for an explicit construction in the case of degenerate horizons, see Khuri--Weinstein \cite[Section 3]{KhuriWeinstein2016}. A typical exhaustion argument then yields the existence of a unique axisymmetric harmonic map $\Phi:\mathbb{R}^3 \setminus\Gamma\rightarrow\mathbb{H}^2_{\mathbb{C}}$
which is asymptotic to $\Phi_0$ in the sense that
\begin{equation}
 \sup_{\R^3\setminus\Gamma}
 d_{\mathbb{H}^2_{\mathbb{C}}}(\Phi,\Phi_0)<\infty,
 \quad
 d_{\mathbb{H}^2_{\mathbb{C}}}(\Phi,\Phi_0)\rightarrow0
 \quad\text{as }r\rightarrow\infty.
\end{equation}
We refer to the survey \cite{Weinstein} for more details, and record these conclusions in the following result.

\begin{proposition}
\label{exi1} 
Given a rod data set with arbitrary horizon, axis, and potential constant configuration with model map $\Phi_0$,
there is a unique axisymmetric harmonic map $\Phi:\R^3\setminus\Gamma\rightarrow\mathbb{H}^2_{\mathbb{C}}$
which is asymptotic to the model map.
\end{proposition}

\begin{remark}\label{rem:futureVerification}
In a forthcoming work \cite{HKWXprog}, the authors will show that each of the harmonic
maps arising from Proposition \ref{exi1} satisfy the conditions
\textup{(A1)}--\textup{(A4)}, and are hence stationary electrovacuum singular harmonic maps.
In particular, this will show that the analyticity condition assumed at the axes, is not an
additional restriction on the class of maps, but is rather a natural consequence of the regularity
theory. 
\end{remark}

\section{Weyl Conformal Factor Global Upper Bound}
\label{sec3}

The desired global upper bound for the Weyl conformal factor will be achieved with the maximum
principle, applied to a nonlinear differential inequality that generalizes \cite[Lemma 3.1]{HKWXvacuum} in the vacuum setting.
We will first compute the 2-dimensional Laplacian of $\alpha$. It will be convenient to introduce the notation
\begin{equation}
 \cA=|\grad U|^2,\qquad
 \cT=e^{4u}|\Th|^2,
 \qquad
 \cE=e^{2u}\bigl(|\grad\chi|^2+|\grad\psi|^2\bigr),
\end{equation}
and observe that the harmonic map equations imply
\begin{equation}\label{eq:Uequation}
 \Delta U=2\cT+\cE.
\end{equation}
Moreover, note that if $\delta=d\rho^2+dz^2+\rho^2 d\phi^2$ is the flat metric on $\mathbb{R}^3$ in cylindrical coordinates, then the stress--energy tensor for the harmonic map is given by
\[
 T_\Phi
 =\Phi^*h_{\CH}
  -\frac12|d\Phi|_{\delta,h_{\CH}}^2\,\delta,
\]
with components
\begin{equation}
(T_{\Phi})_{\rho\rho}=\frac{1}{2}\left(|\Phi_{\rho}|^2_{h_{\mathbb{H}^2_{\mathbb{C}}}}-|\Phi_{z}|^2_{h_{\mathbb{H}^2_{\mathbb{C}}}}\right)=-(T_{\Phi})_{zz},\quad\quad\quad\quad (T_{\Phi})_{\rho z}=\langle\Phi_{\rho},\Phi_{z}\rangle_{h_{\mathbb{H}^2_{\mathbb{C}}}}.
\end{equation}
Here the target inner products are computed with respect to
\eqref{eq:CHmetric}, for instance
\[
 |\Phi_\rho|_{h_{\CH}}^2
 =u_\rho^2+e^{4u}\Th_\rho^2
  +e^{2u}(\chi_\rho^2+\psi_\rho^2).
\]

\begin{lemma}\label{lem:alphaIdentity}
Away from $\Gamma$, the Weyl conformal factor satisfies
\begin{equation}\label{eq:alphaIdentity}
 \Dtwo\alpha=-\cA+3\cT+\cE.
\end{equation}
\end{lemma}

\begin{proof}
Set
\begin{equation}
 \xi=\alpha-2u-\ln\rho.
\end{equation}
Substituting $U=u+\ln\rho$ into
\eqref{eq:alphaquad1}--\eqref{eq:alphaquad2} produces
\begin{align}
 \xi_\rho
 =\rho\left[|\Phi_\rho|_{h_{\CH}}^2-|\Phi_z|_{h_{\CH}}^2\right]=2\rho (T_{\Phi})_{\rho\rho},
 \qquad\quad
 \xi_z
 =2\rho\,\langle\Phi_\rho,\Phi_z\rangle_{h_{\CH}}=2\rho(T_{\Phi})_{\rho z}.
\end{align}
Next differentiate these two equations, denoting by $D$ the
pullback connection of $\CH$, to obtain
\begin{align*}
 \Dtwo \xi
 &=|\Phi_\rho|_{h_{\CH}}^2-|\Phi_z|_{h_{\CH}}^2
  +\rho\,\partial_\rho
     (|\Phi_\rho|_{h_{\CH}}^2-|\Phi_z|_{h_{\CH}}^2)
  +2\rho\,\partial_z\langle\Phi_\rho,\Phi_z\rangle_{h_{\CH}}\\
 &=|\Phi_\rho|_{h_{\CH}}^2-|\Phi_z|_{h_{\CH}}^2
  +2\rho\left\langle
       \Phi_\rho,D_\rho\Phi_\rho+D_z\Phi_z
     \right\rangle_{h_{\CH}}.
\end{align*}
In covariant form, the axisymmetric harmonic map equations are
\[
 D_\rho\Phi_\rho+D_z\Phi_z+\rho^{-1}\Phi_\rho=0,
\]
and thus
\begin{equation}
 \Dtwo \xi=-|\Phi_\rho|_{h_{\CH}}^2-|\Phi_z|_{h_{\CH}}^2
 =-|\grad u|^2-\cT-\cE.
\end{equation}
Moreover,
\[
 \Dtwo u=2\cT+\cE-\frac{u_\rho}{\rho},
 \qquad
 \Dtwo\ln\rho=-\frac1{\rho^2}.
\]
Since $\alpha=\xi+2u+\ln\rho$ and
$U=u+\ln\rho$, it follows that
\begin{align*}
 \Dtwo\alpha
 &=-|\grad u|^2-\cT-\cE+4\cT+2\cE
   -\frac{2u_\rho}{\rho}-\frac1{\rho^2}\\
 &=-|\grad U|^2+3\cT+\cE,
\end{align*}
which gives \eqref{eq:alphaIdentity}.
\end{proof}

We are now ready to establish the nonlinear differential inequality for the Weyl conformal factor.

\begin{proposition}\label{prop:alphaIneq}
Away from $\Gamma$, the Weyl conformal factor satisfies
\begin{equation}\label{eq:alphaIneq}
 \Dtwo\alpha+\frac{|\grad\alpha|}{\rho}
 \geq 2\cT\geq0.
\end{equation}
\end{proposition}

\begin{proof}
Consider the complex valued functions
\[
 \zeta_U=U_\rho+iU_z,
 \quad \zeta_0=e^{2u}(\Th_\rho+i\Th_z),
 \quad \zeta_1=e^u(\chi_\rho+i\chi_z),
 \quad \zeta_2=e^u(\psi_\rho+i\psi_z),
\]
and observe that the quadrature equations imply
\begin{equation}\label{eq:complexSquareExact}
 \frac{|\grad\alpha|}{\rho}
 =\left|\zeta_U^2+\zeta_0^2+\zeta_1^2+\zeta_2^2\right|.
\end{equation}
Indeed, the real and imaginary parts of the complex number on the right-hand side of
\eqref{eq:complexSquareExact} are exactly $\alpha_\rho/\rho$ and
$\alpha_z/\rho$, by \eqref{eq:alphaquad1}--\eqref{eq:alphaquad2}. Furthermore, the triangle inequality gives
\[
 \left|\zeta_U^2+\sum_{j=0}^2 \zeta_j^2\right|
 \geq |\zeta_U|^2-\sum_{j=0}^2|\zeta_j|^2
 =\cA-\cT-\cE.
\]
Applying Lemma \ref{lem:alphaIdentity}, we find that
\begin{equation}
    \Dtwo \alpha +\frac{|\nabla\alpha|}{\rho}\geq -\cA+3\cT+\cE +(\cA-\cT-\cE)\geq 2\mathcal{T}.
\end{equation}
This yields \eqref{eq:alphaIneq}.
\end{proof}

We will now obtain the global upper bound for $\alpha$, by applying the maximum principle on a sequence of exhaustion domains. The first step is to estimate the asymptotic profile at each puncture.

\begin{lemma}\label{lem:tangentProfileMax}
Let $p_i$ be a degenerate horizon puncture, and let $\bar{\alpha}_i$ be the function provided by condition 
\textup{(A3)}. Then
\begin{equation}\label{eq:tangentProfileMax}
 \overline\alpha_i(\theta)
 \leq\max\{\mathbf{b}_i,\mathbf{b}_{i+1}\},
 \qquad 0\leq\theta\leq\pi.
\end{equation}
\end{lemma}

\begin{proof}
Insert the asymptotic expansion of $\alpha$ near $p_i$, as given by condition \textup{(A3)}, into the differential inequality of Proposition~\ref{prop:alphaIneq}, multiply by $r_i^2$ and take the limit as $r_i \rightarrow 0$ to find
\begin{equation}\label{eq:tangentODE}
 \overline\alpha_i''(\theta)
 +\frac{|\overline\alpha_i'(\theta_i)|}{\sin\theta_i}\geq0,
 \qquad 0<\theta_i<\pi.
\end{equation}
We will show that a function satisfying \eqref{eq:tangentODE} cannot exceed the
larger endpoint value.

In what follows, for convenience, we will remove the $i$ subscript from $\theta_i$, and set $y=\overline\alpha_i'$. Suppose that $y$ is positive at some interior point.  On the
connected component of $\{y>0\}$ containing that point,
\eqref{eq:tangentODE} reads
\[
 y'+\frac{y}{\sin\theta}\geq0.
\]
Since
\[
 \frac{d}{d\theta}\ln\tan\frac\theta2=\frac1{\sin\theta},
\]
it follows that
\begin{equation}\label{eq:integratingTangent}
 \frac{d}{d\theta}\left(y\,\tan\frac\theta2\right)\geq0.
\end{equation}
Let $(a,b)$ be that connected component.  If $b<\pi$, continuity gives
$y(b)=0$, while \eqref{eq:integratingTangent} says that
$y(\theta)\tan(\theta/2)$ is nondecreasing and is already positive at the
chosen point, yielding a contradiction.  Hence $b=\pi$, and $y>0$ on the entire open
interval from the chosen point to $\pi$.  The function $\overline{\alpha}_i$ is therefore strictly
increasing there.

If \eqref{eq:tangentProfileMax} failed, continuity and the endpoint values
\eqref{eq:tangentEndpoints} would give an interior point $\theta_*$, at which $\overline{\alpha}_i$ takes a
value larger than both endpoint values.  By the mean value theorem there
would be a point to the left of $\theta_*$ where $y>0$.  The preceding
paragraph then shows that $y>0$ from that point all the way to $\pi$, so
$\overline\alpha_i(\pi)>\overline\alpha_i(\theta_*)$, contradicting the
choice of $\theta_*$.  We conclude that \eqref{eq:tangentProfileMax} must hold.
\end{proof}

\begin{proposition}\label{prop:alphaUpper}
Set
\begin{equation}\label{eq:defectMaximum}
 M=\max_{1\leq j\leq N+1}\mathbf{b}_j.
\end{equation}
Then $M\geq0$ and
\begin{equation}\label{eq:alphaUpper}
 \alpha(\rho,z)\leq M
 \qquad\text{throughout the open orbit half-plane }\rho>0.
\end{equation}
\end{proposition}

\begin{proof}
Observe that the semi-infinite axis normalization \eqref{eq:outerDefects} gives $M\geq0$.  To establish the
upper bound for $\alpha$, we will use the maximum principle on a collection of exhaustion domains. 
Let $Q$ be the finite set consisting
of all poles of nondegenerate horizon rods and all degenerate punctures.  Choose
$\sigma>0$, $R\gg1$, and $0<\eta<\sigma/2$, and define
\begin{equation}\label{domain1}
 \Omega_{\eta,\sigma,R}
 =\{(\rho,z)\mid \rho>\eta,\ \rho^2+z^2<R^2,
 \ r_q>\sigma\text{ for every }q\in Q\},
\end{equation}
where $r_q$ denotes the Euclidean distance to point $q\in\mathbb{R}^3$. The value $R$ should be chosen sufficiently large so that all points of $Q$ lie within the sphere of radius $R/2$.
On this domain \eqref{domain1} set
\[
 B=\begin{cases}
 \dfrac{\nabla\alpha}{\rho|\nabla\alpha|},&\nabla\alpha\neq0,\\[1.1ex]
 0,&\nabla\alpha=0,
 \end{cases}
\]
and note that $B$ is bounded and measurable, with $|B|\leq\eta^{-1}$.
By Proposition~\ref{prop:alphaIneq} we have
\[
 \Delta_{\rho,z}\alpha+B\cdot\nabla\alpha\geq0,
\]
and hence the maximum principle yields
\begin{equation}\label{eq:MPexhaust}
 \sup_{\Omega_{\eta,\sigma,R}}\alpha
 \leq\sup_{\partial\Omega_{\eta,\sigma,R}}\alpha.
\end{equation}

We now estimate the boundary pieces.  At a degenerate puncture,
condition~\textup{(A3)} and Lemma~\ref{lem:tangentProfileMax} provide a modulus
$\varepsilon_q(\sigma)\to0$ such that
\begin{equation}
 \sup_{r_q=\sigma}\alpha
 \leq M+\varepsilon_q(\sigma).
\end{equation}
At a nondegenerate lower or upper pole,
\eqref{eq:polealphaLower}--\eqref{eq:polealphaUpper} and
$\ln\sin(\theta/2),\ln\cos(\theta/2)\leq0$ imply
\begin{equation}
 \sup_{r_q=\sigma}\alpha\leq M+C_q\sigma^{\gamma_q},
\end{equation}
for some $\gamma_q >0$. Thus, since $Q$ is finite, the error in all these small circle estimates tends uniformly
to zero with $\sigma$.
On portions of the boundary with $\rho=\eta$ that project to compact subsets of an 
axis rod, Proposition~\ref{prop:axisFactor} produces
\[
 \alpha 
 \leq M+C_{\sigma,R}\eta^2,
\]
while on portions of the boundary projecting to the interior of a nondegenerate horizon rod,
\eqref{eq:horizonexpansion} gives
\[
 \alpha\leq\ln\eta+C_{\sigma,R}.
\]
Moreover, on the outer semicircle $\rho^2 +z^2 =R^2$, it follows from \eqref{eq:infinity} that
\begin{equation}
|\alpha|\leq CR^{-2}\leq M+CR^{-2}, 
\end{equation}
since $M\geq0$.

Fix $(\rho_0,z_0)$ arbitrarily with $\rho_0>0$.  For all sufficiently large $R$ and sufficiently small $\sigma,\eta$, this point belongs to $\Omega_{\eta,\sigma,R}$.  In \eqref{eq:MPexhaust}, first let
$\eta\rightarrow0$ with $\sigma,R$ fixed, then let $\sigma\rightarrow0$, and
finally let $R\to\infty$.  The horizon term tends to $-\infty$ and all error
terms vanish. It follows that $\alpha(\rho_0,z_0)\leq M$.
\end{proof}

\section{Analysis at the Axis}
\label{sec4}

In this section we will first establish technical properties concerning the Weyl conformal factor expansion at axis rods,
and will then prove a rigidity statement for analytic functions satisfying the differential inequality \eqref{eq:alphaIneq}.

\subsection{Axis factorization}
Fix a compact subinterval $I\subset\Gamma_j$ of an axis rod and assume that target space isometry
\eqref{eq:Heisenberg} has been applied so that all potential constants vanish,
\begin{equation}\label{eq:axisGauge}
 v=\chi=\psi=0\qquad\text{on the rod.}
\end{equation}
By \eqref{eq:axisAnalyticAssumption} of condition (A1), there exist real analytic functions $F,H,K,G$ on
a neighborhood of the subinterval such that
\begin{equation}
 U=F(s,z),\qquad
 v=s^2G(s,z),\qquad
 \chi=sH(s,z),\qquad
 \psi=sK(s,z).
\end{equation}
where $s=\rho^2$. We will also set
\[
 f(z)=F(0,z),\qquad h(z)=H(0,z),\qquad k(z)=K(0,z).
\]

\begin{proposition}\label{prop:axisFactor}
Let $I\subset\Gamma_j$ be an axis rod subinterval as described above. Then there exists a neighborhood
$\mathcal{O}_I\subset\mathbb{R}^3$, and a real analytic function $A$, such that on $\mathcal{O}_I$,
\begin{equation}\label{eq:alphaAxisAnalytic}
 \alpha(\rho,z)=\mathbf{b}_j+\rho^2A(\rho^2,z).
\end{equation}
Moreover, on $\mathcal{O}_I \cap\Gamma_j$ it holds that
\begin{equation}\label{lem:alphaNormal}
 \alpha_{\rho\rho}(0,z)
 =-f'(z)^2+4e^{2f(z)}\bigl(h(z)^2+k(z)^2\bigr).
\end{equation}
\end{proposition}

\begin{proof} 
For $s>0$, write
$\widehat\alpha(s,z)=\alpha(\sqrt{s},z)$, and define the analytic functions
\[
 \begin{aligned}
 a&=H+sH_s, & b&=K+sK_s,\\
 c&=2G+s(G_s-KH_s+HK_s),
 &d&=G_z-KH_z+HK_z.
 \end{aligned}
\]
Then
\begin{equation}
 U_\rho=2\rho F_s,\qquad  \chi_\rho=2\rho a,\qquad \psi_\rho=2\rho b,\qquad  \Theta_\rho=2\rho^3c,
\end{equation} 
\begin{equation}
 U_z=F_z,\qquad
 \chi_z=sH_z,\qquad
  \psi_z=sK_z,\qquad
 \Theta_z=\rho^4d.
\end{equation}
Since $e^{2u}=s^{-1}e^{2F}$ and $e^{4u}=s^{-2}e^{4F}$, the radial Weyl
quadrature equation \eqref{eq:alphaquad1} yields
\begin{equation}
 \widehat\alpha_s=2sF_s^2-\frac12F_z^2
 +2se^{4F}c^2-\frac12s^2e^{4F}d^2
 +2e^{2F}(a^2+b^2)
 -\frac12se^{2F}(H_z^2+K_z^2),
\end{equation}
which is real analytic in $(s,z)$.  In particular,
there exists an analytic function $A$ such that, in a neighborhood of $I$,
\[
 \widehat\alpha(s,z)=\mathbf{b}_j
 +\int_0^s\widehat\alpha_s(t,z)\,dt
 =\mathbf{b}_j+sA(s,z),
\]
which gives \eqref{eq:alphaAxisAnalytic}. Next, observe that
\begin{equation}
 U_\rho=O(\rho),\qquad \chi_\rho=2\rho h(z)+O(\rho^3),\qquad \psi_\rho=2\rho k(z)+O(\rho^3),\qquad \Th_\rho=O(\rho^3),
\end{equation}
\begin{equation}
U_z=f'(z)+O(\rho^2),\qquad
\chi_z=O(\rho^2),\qquad
\psi_z=O(\rho^2),\qquad
\Th_z=O(\rho^4).
\end{equation}
These formulas, together with $e^{2u}=\rho^{-2}e^{2U}$ and $e^{4u}=\rho^{-4}e^{4U}$, inserted into
\eqref{eq:alphaquad1} produce
\[
 \alpha_\rho
 =\rho\left[-f'(z)^2+4e^{2f(z)}(h(z)^2+k(z)^2)\right]+O(\rho^3).
\]
Differentiating with respect to $\rho$ and evaluating at $\rho=0$ gives \eqref{lem:alphaNormal}.
\end{proof}


\subsection{Boundary rigidity}
The next result gives a local boundary rigidity statement for a certain class of functions.
This will be applied to the Weyl conformal factor in the next subsection, and is a central mechanism
in the strict inequality argument for the angle defects.

\begin{lemma}\label{thm:boundaryRigidity}
Let $a=a(\rho,z)$ be real analytic near $(0,0)$, is even in $\rho$, and suppose
that in a half neighborhood $\rho\geq0$, it holds that
\begin{equation}
 a\leq0,\qquad a(0,z)=0.
\end{equation}
Assume further that for $\rho>0$,
\begin{equation}\label{eq:boundaryIneq}
 a_{\rho\rho}+a_{zz}+\frac{|\grad a|}{\rho}\geq0.
\end{equation}
If
\begin{equation}
 a_{\rho\rho}(0,0)=0,
\end{equation}
then $a\equiv0$ in a neighborhood of the origin.
\end{lemma}

\begin{proof}
For convenience we set $V=-a$. Observe that this function satisfies
\begin{equation}\label{eq:Vineq}
 V\geq0 \text{ for }\rho\geq 0,\qquad V(0,z)=0,
 \qquad
 \Dtwo V-\frac{|\grad V|}{\rho}\leq0 \text{ for }\rho>0.
\end{equation}
Because $V$ is analytic, is even in $\rho$, and vanishes on $\rho=0$, there is
an analytic even function $P$ such that
\begin{equation}\label{eq:Vfactor}
 V(\rho,z)=\rho^2 P(\rho,z).
\end{equation}
Moreover $P\geq0$ for $\rho\geq0$, and
$P(0,0)=\tfrac12V_{\rho\rho}(0,0)=0$ since $a_{\rho\rho}(0,0)=0$.

Assume, by way of contradiction, that $P$ is not identically zero.  Let $P_d$ be
the first nonzero homogeneous term in its Taylor expansion at the origin.
Then $d\geq1$, $P_d$ is even in $\rho$, and
\begin{equation}
 P_d(\rho,z)\geq0\qquad\text{for }\rho\geq0.
\end{equation}
Indeed, this follows by restricting $P$ to every ray and taking the first
nonzero coefficient.  The first homogeneous term of $V$ is
\[
 V_n(\rho,z)=\rho^2P_d(\rho,z),\qquad n=d+2>2.
\]
Passing to the leading homogeneous order in the differential inequality of \eqref{eq:Vineq} gives
\begin{equation}\label{eq:homIneqCartesian}
 \Dtwo V_n-\frac{|\grad V_n|}{\rho}\leq0
 \qquad\text{for } \rho>0.
\end{equation}
Next, introduce polar coordinates
\[
 \rho=r\sin\theta,
 \qquad z=r\cos\theta,
 \qquad 0<\theta<\pi,
\]
and write
\begin{equation}
 V_n=r^n p(\theta).
\end{equation}
Then $p\geq0$, and the factor $\rho^2$ shows that $p$ vanishes to order at
least two at both endpoints.  Equation \eqref{eq:homIneqCartesian} then becomes
\begin{equation}\label{eq:qIneq}
 p''+n^2p-\frac{\sqrt{n^2p^2+p'^2}}{\sin\theta}\leq0.
\end{equation}

We first show that $p>0$ on $(0,\pi)$.  On any component where $p>0$, define
$\varsigma\in(-\pi/2,\pi/2)$ by
\begin{equation}\label{eq:betaDef}
 \tan\varsigma=\frac{p'}{np}.
\end{equation}
A direct substitution into \eqref{eq:qIneq} yields
\begin{equation}\label{eq:betaIneq}
 \varsigma'+n\leq\frac{\cos\varsigma}{\sin\theta}.
\end{equation}
If $p$ had an interior zero at $\theta_0$, nonnegativity and analyticity would
give, on one side of the zero,
$p(\theta)=c|\theta-\theta_0|^m(1+o(1))$ with $m\geq2$ and $c>0$.  Approaching from the
right gives
\[
 \varsigma'\longrightarrow-\frac nm,
 \qquad
 \frac{\cos\varsigma}{\sin\theta}\longrightarrow0.
\]
Taking the limit in \eqref{eq:betaIneq} would imply
$n(1-1/m)\leq0$, which is impossible.  Hence $p>0$ in $(0,\pi)$. 

Let $m_-$ be the order of vanishing at $\theta=0$.  Then
$p(\theta)=c\theta^{m_-}(1+o(1))$, and \eqref{eq:betaDef} gives
\[
 \varsigma'\longrightarrow-\frac{n}{m_-},
 \qquad
 \frac{\cos\varsigma}{\sin\theta}\longrightarrow\frac{n}{m_-}.
\]
Equation \eqref{eq:betaIneq} therefore implies $m_-\leq2$.  Since the factor
$\rho^2$ gives $m_-\geq2$, one has $m_-=2$.  An identical argument at
$\theta=\pi$ shows that the upper endpoint order is also two.

Now, reparameterize $\theta\in(0,\pi)$ by $t\in (-\infty,\infty)$ and define a function $y(t)$ as follows
\begin{equation}
 t=\ln\tan\frac\theta2,
 \qquad \sin\theta=\operatorname{sech}t,
 \qquad y(t)=\sin\varsigma(\theta(t)).
\end{equation}
Then multiplying \eqref{eq:betaIneq} by
$\sin\theta\cos\varsigma$ produces
\begin{equation}\label{eq:yIneq}
 y'\leq1-y^2-n\operatorname{sech}t\sqrt{1-y^2},
\end{equation}
where $y'=\frac{dy}{dt}$. We will make a comparison with the function
\begin{equation}
 y_0(t)=-\tanh t.
\end{equation}
Observe that replacing $y$ with $y_0$ in the right-hand side of \eqref{eq:yIneq} yields
$(1-n)\operatorname{sech}^2t$, whereas
$y_0'=-\operatorname{sech}^2t$. Thus, since $n>2$ we obtain
\begin{equation}\label{eq:strictSubsolution}
 1-y_0^2-n\operatorname{sech}t\sqrt{1-y_0^2}<y_0'.
\end{equation}
Moreover, since the endpoint vanishing for $p$ is of order two it follows that
\begin{align}
 y(t)&=1-\frac{n^2}{2}e^{2t}+o(e^{2t})
 &&\text{as }t\to-\infty,\label{eq:yminus}\\
 y(t)&=-1+\frac{n^2}{2}e^{-2t}+o(e^{-2t})
 &&\text{as }t\to+\infty.\label{eq:yplus}
\end{align}
On the other hand,
\begin{align*}
 y_0(t)&=1-2e^{2t}+o(e^{2t})&&\text{as }t\to-\infty,\\
 y_0(t)&=-1+2e^{-2t}+o(e^{-2t})&&\text{as }t\to+\infty.
\end{align*}
Since $n>2$, \eqref{eq:yminus} gives $y<y_0$ for all sufficiently negative
$t$.  At a first finite time contact from below we would have
$(y-y_0)'\geq0$, while \eqref{eq:yIneq} and
\eqref{eq:strictSubsolution} give $(y-y_0)'<0$.  Therefore $y<y_0$ for every
finite $t$. However, \eqref{eq:yplus} ensures that $y>y_0$ for all sufficiently positive
$t$, yielding a contradiction.

We conclude that no nonzero first homogeneous term $P_d$ exists in the Taylor expansion of $P$.  Real analyticity then implies
that $P\equiv0$ near the origin, and hence 
$a\equiv0$ in the same neighborhood.
\end{proof}

\subsection{Rigidity at a maximal axis rod}
We will now apply Lemma \ref{thm:boundaryRigidity} at an axis rod having the largest logarithmic angle defect.

\begin{proposition}\label{prop:maxRodRigidity}
Let $M$ be defined by \eqref{eq:defectMaximum}.  If a bounded axis rod
$\Gamma_j$ satisfies
\begin{equation}\label{eq:maxRod}
 \mathbf{b}_j=M,
\end{equation}
then
\begin{equation}\label{eq:alphaZero}
 \alpha\equiv0 \qquad \text{ throughout the orbit space half-plane},
\end{equation}
and in particular $M=0$.
\end{proposition}

\begin{proof}
We first observe that the axis trace $f(z)=U(0,z)$ on $\Gamma_j$, achieves an interior maximum.
This is due to the fact that $f$ is continuous in the open interval and tends to $-\infty$ at both endpoints.
Indeed, this asymptotic behavior follows from \eqref{eq:poleU} for an endpoint belonging to a nondegenerate horizon, and follows from
\eqref{eq:tangentUConvergence} for an endpoint that is a degenerate horizon puncture.  
We may therefore choose an interior maximum point and translate it to $z=0$, so that $f'(0)=0$. We will also choose the target space
gauge \eqref{eq:axisGauge}.

Set
\[
 a=\alpha-M.
\]
By Proposition~\ref{prop:alphaUpper}, $a\leq0$ in the open half-plane.  By \eqref{eq:alphaAxisAnalytic}
and \eqref{eq:maxRod}, $a(0,z)=a_{\rho}(0,z)=0$ near the chosen
point, and since $\rho\mapsto a(\rho,0)$ has a one-sided maximum at $\rho=0$ it follows that
\begin{equation}
 \alpha_{\rho\rho}(0,0)=a_{\rho\rho}(0,0)\leq0.
\end{equation}
On the other hand, \eqref{lem:alphaNormal} and $f'(0)=0$ give
\[
 \alpha_{\rho\rho}(0,0)
 =4e^{2f(0)}\bigl(h(0)^2+k(0)^2\bigr)\geq0.
\]
Therefore
\begin{equation}
 h(0)=k(0)=0,
 \qquad
 a_{\rho\rho}(0,0)=0.
\end{equation}

By Proposition~\ref{prop:axisFactor}, $a$ is a real-analytic function of $(\rho^2,z)$ near
the axis, and satisfies
\[
 a_{\rho\rho}+a_{zz}+\frac{|\nabla a|}{\rho}\geq0
\]
by Proposition~\ref{prop:alphaIneq}.  The boundary rigidity Lemma
\ref{thm:boundaryRigidity} then applies and gives $a\equiv0$ near the chosen point.
Since $a$ is real analytic on the connected open orbit space half plane, analytic
continuation yields
\[
 \alpha\equiv M.
\]
Lastly, the asymptotic normalization \eqref{eq:infinity} forces $M=0$.
\end{proof}

\section{Proof of the Main Theorem}\label{sec:equality}

The main theorem involves a dichotomy statement, between the balanced and unbalanced classes of solutions. 
These two classes are characterized by the vanishing or nonvanishing of $\alpha$, respectively.

\subsection{The Majumdar--Papapetrou case}
In this subsection we will classify all solutions for which $\alpha$ vanishes.

\begin{proposition}\label{prop:conformastatic}
Assume that $\alpha\equiv0$ on the orbit space half-plane.  Then the following global properties of the harmonic map and its associated electrovacuum spacetime hold.
\begin{enumerate}[label=\textup{(\roman*)},leftmargin=2.5em]
\item The twist 1-form is trivial, and the two remaining pieces of the renormalized energy density agree
\begin{equation}\label{eq:TzeroAE}
 \Th\equiv0,
 \qquad
 |\grad U|^2=e^{2u}\bigl(|\grad\chi|^2+|\grad\psi|^2\bigr).
\end{equation}

\item
There is a constant unit vector
$\eta_0\in S^1\subset\R^2$ such that, with $\mathcal{U}=e^{-U}$,
one has
\begin{equation}\label{eq:pHarmonicZ}
 \Delta \mathcal{U}=0,
 \qquad
 d(\chi,\psi)=\eta_0\,\rho\starTwo d\mathcal{U}.
\end{equation}

\item The spacetime metric \eqref{eq:WP} takes the form
\begin{equation}\label{eq:conformastaticMetric}
 \mathbf g=-\mathcal{U}^{-2}dt^2+\mathcal{U}^2
 \bigl(d\rho^2+d z^2+\rho^2 d\phi^2\bigr).
\end{equation}

\item
After a Heisenberg translation, the twist potential vanishes $v\equiv 0$, and the Maxwell field is a constant duality rotation of
$d(\mathcal{U}^{-1}d t)$, in particular
\[
 d\chi=\cos\gamma\,\rho\starTwo d \mathcal{U},
 \qquad
 d\psi=\sin\gamma\,\rho\starTwo d \mathcal{U},
\]
for some constant $\gamma\in[0,2\pi)$.
\end{enumerate}
\end{proposition}

\begin{proof}
Since $\alpha$ is constant, \eqref{eq:alphaIneq} implies that $\cT=0$, and hence
$\Th=0$. According to \eqref{eq:alphaIdentity} we then have $\cA=\cE$, which yields 
\eqref{eq:TzeroAE}. Next, observe that the two quadrature equations \eqref{eq:alphaquad1}--\eqref{eq:alphaquad2}, with
$\alpha=0$ and $\Th=0$, give
\begin{align}
 U_\rho^2-U_z^2
 +e^{2u}\bigl(\chi_\rho^2-\chi_z^2+
 \psi_\rho^2-\psi_z^2\bigr)&=0,
 \label{eq:stress1}\\
 U_\rho U_z+e^{2u}\bigl(\chi_\rho\chi_z+
 \psi_\rho\psi_z\bigr)&=0.
 \label{eq:stress2}
\end{align}
Set
\[
 X=e^u(\chi_\rho,\psi_\rho),
 \qquad
 Y=e^u(\chi_z,\psi_z),
 \qquad a=U_\rho,
 \qquad b=U_z.
\]
Then combining \eqref{eq:TzeroAE}, \eqref{eq:stress1}, and
\eqref{eq:stress2} produces
\begin{equation}
 |X|^2=b^2,
 \qquad |Y|^2=a^2,
 \qquad X\cdot Y=-ab.
\end{equation}
Thus, equality holds in Cauchy--Schwarz, and it follows that on any open set where $\grad U\neq0$,
there is a unit vector field $\eta$ such that
\begin{equation}\label{eq:XYeta}
 X=b\eta,
 \qquad Y=-a\eta,\qquad \eta=\frac{bX-aY}{a^2 +b^2}.
\end{equation}
Since $e^{-u}=\rho e^{-U}=\rho \mathcal{U}$, equations \eqref{eq:XYeta} may be rewritten as
\begin{equation}\label{eq:dZomega}
 d Z=\eta\,\omega,
 \qquad
 Z=(\chi,\psi),
 \qquad
 \omega=\rho(-\mathcal{U}_z d\rho+\mathcal{U}_\rho d z)
       =\rho\starTwo d \mathcal{U}.
\end{equation}

From \eqref{eq:Uequation} and \eqref{eq:TzeroAE},
$\Delta U=|\grad U|^2$, and hence
\begin{equation}
 \Delta \mathcal{U}=e^{-U}\bigl(-\Delta U+|\grad U|^2\bigr)=0.
\end{equation}
Moreover, a direct calculation shows that
\begin{equation}
 d\omega=\rho(\Delta \mathcal{U})\,d\rho\wedge d z=0,
\end{equation}
and taking an exterior derivative in \eqref{eq:dZomega} gives
\begin{equation}\label{eq:detaomega}
 d\eta\wedge\omega=0.
\end{equation}
When $\Th=0$, the third and fourth equations of \eqref{eq:harmu} become
\begin{equation}
 \operatorname{div}(e^{2u}\grad Z)=0 \qquad \Longleftrightarrow \qquad
 d\bigl(\rho e^{2u}\starTwo d Z\bigr)=0.
\end{equation}
Using \eqref{eq:dZomega}, $\rho e^{2u}=1/(\rho \mathcal{U}^2)$, and
$\starTwo^2=-1$ on 1-forms, we obtain
\[
 \rho\, e^{2u}\starTwo d Z
 =-\eta\, \mathcal{U}^{-2}d \mathcal{U}=\eta\,d(\mathcal{U}^{-1}),
\]
and thus
\begin{equation}\label{eq:detadp}
 d\eta\wedge d(\mathcal{U}^{-1})=0.
\end{equation}
Where $d \mathcal{U}\neq0$, the 1-forms $d \mathcal{U}$ and
$\omega=\rho\starTwo d \mathcal{U}$ are linearly independent.  Equations
\eqref{eq:detaomega} and \eqref{eq:detadp} then imply that $d\eta=0$ on such a region.  Choose
a nonempty connected component of the set on which $d \mathcal{U}\neq0$; note that $\mathcal{U}$ is not globally constant because the
solution contains a horizon.  On this set, $\eta=\eta_0$ is constant and
\eqref{eq:dZomega} holds with $\eta_0$. Since $dZ-\eta_0 \omega=0$ on an open set, and the left-hand side is an analytic 1-form,
it follows that this holds globally on the orbit space half-plane, which yields
\eqref{eq:pHarmonicZ}.

To establish (iii), note that the
quadrature equation \eqref{eq:wquad} gives $d w=0$; asymptotic flatness fixes $w=0$.
Substituting into \eqref{eq:WP} then produces \eqref{eq:conformastaticMetric}.

We will now prove (iv). Let $\eta_0^\perp=(-\eta_0^2,\eta_0^1)$. Observe that, throughout the half-plane,
we then have
\[
 d(\eta_0^\perp\cdot Z)
 =\eta_0^\perp\cdot d Z
 =(\eta_0^\perp\cdot\eta_0)\omega=0.
\]
Hence
\[
 \eta_0^\perp\cdot Z=c_\perp
\]
for some constant $c_\perp$, and it follows that
\[
 Z=q\eta_0+c_\perp\eta_0^\perp,
 \qquad q=\eta_0\cdot Z.
\]
Therefore, the image of $Z$ is contained in an affine line parallel to
$\eta_0$. Choose the constant Heisenberg translation \eqref{eq:Heisenberg} with
$(a_2,a_3)=-c_\perp\eta_0^\perp$,
namely
\[
 \widetilde\chi=\chi+a_2,\qquad
 \widetilde\psi=\psi+a_3,\qquad
 \widetilde v=v+a_3\chi-a_2\psi+a_1.
\]
This transformation preserves
\[
 \Theta=d v-\psi\,d\chi+\chi\,d\psi,
\]
and for the translated electromagnetic potential
$\widetilde Z=(\widetilde\chi,\widetilde\psi)$ one has
\[
 \widetilde Z=q\eta_0.
\]
Therefore
\[
 -\widetilde\psi\,d\widetilde\chi
 +\widetilde\chi\,d\widetilde\psi=0.
\]
Since $\widetilde\Theta=\Theta=0$, it follows that
$d\widetilde v=0$.  The remaining translation parameter
$a_1$ may then be chosen so that $\widetilde v=0$. In what follows, we will denote $\widetilde{v}$, $\widetilde{\chi}$, $\widetilde{\psi}$ by $v$, $\chi$, $\psi$.

Define
\[
 \mathbf F_0=d(\mathcal{U}^{-1}d t)
             =\mathcal{U}^{-2}d t\wedge d \mathcal{U},
\]
and observe that a direct calculation in the metric
\eqref{eq:conformastaticMetric} produces
\[
 {*}_{\mathbf g}\mathbf F_0
 =\rho\,d\phi\wedge\starTwo d \mathcal{U}.
\]
It follows that
\[
\iota_\eta {*}_{\mathbf g}\mathbf F_0
 =\rho\starTwo d \mathcal{U},\qquad
 \iota_\eta\mathbf F_0=0,\qquad \eta=\partial_{\phi}.
\]
Since ${*}_{\mathbf g}^2=-1$ on 2-forms, the axial potential pairs
associated with $\mathbf F_0$ are
\[
 (d\chi_0,d\psi_0)
 =(\rho\starTwo d \mathcal{U},0).
\]
Note that since $\mathcal{U}$ is harmonic, $\mathbf{F}_0$ satisfies the source-free Maxwell equations $d\mathbf{F}_0=d{*}_{\mathbf g}\mathbf F_0=0$.
Now write the constant unit vector in \eqref{eq:pHarmonicZ} as
$\eta_0=(\cos\gamma,\sin\gamma)$.
Then
\begin{equation}\label{uhb}
 d\chi=\cos\gamma\,\rho\starTwo d \mathcal{U}=\iota_{\eta} {*}_{\mathbf g}\mathbf F,
 \qquad
 d\psi=\sin\gamma\,\rho\starTwo d \mathcal{U}=\iota_{\eta}\mathbf F,
\end{equation}
where $\mathbf{F}$ is the following duality rotation of $\mathbf{F}_0$,
\[
 \mathbf F
 =\cos\gamma\,\mathbf F_0
  +\sin\gamma\,{*}_{\mathbf g}\mathbf F_0.
\]
Indeed, away from the rotation axis, a spacetime $2$-form $\mathbf{F}$ is uniquely determined
by the pair $(\iota_{\eta} {*}_{\mathbf g}\mathbf F,\iota_{\eta}\mathbf F)$, through the formula
\[
 \mathbf F
 =|\eta|_{\mathbf g}^{-2}
 \left[
  \eta^\flat\wedge\iota_\eta\mathbf F
  -{*}_{\mathbf g}
   \bigl(\eta^\flat\wedge
         \iota_\eta {*}_{\mathbf g}\mathbf F\bigr)
 \right].\qedhere
\]
\end{proof}

To complete the classification of solutions with $\alpha=0$, we next find the explicit form for $\mathcal{U}$
and show that nondegenerate horizons cannot occur.

\begin{proposition}\label{lem:noNondegenerate}
Assume that $\alpha\equiv0$ on the orbit space half-plane. Then there are no nondegenerate horizon
rods. Moreover, if $p_1,\ldots,p_N$ denote the degenerate horizon punctures, then
the positive harmonic function in Proposition~\ref{prop:conformastatic} has the form
\begin{equation}
\mathcal{U}=\mathcal{U}_{mp}:=1+\sum_{i=1}^{N}\frac{m_i}{|\mathbf{x}-p_i|}
\end{equation}
for some constants $m_i>0$.
\end{proposition}

\begin{proof}
Assume that a nondegenerate horizon rod exists. On the interior of this rod, \eqref{eq:horizonexpansion} gives
$u=u_0(z)+O_2(\rho^2)$ with $u_0$ finite.  Hence
\[
 \mathcal{U}=e^{-U}=\frac{e^{-u_0(z)}}{\rho}+O_2(\rho),
\]
and thus the leading terms in the axisymmetric Laplacian
are
\[
 \Delta \mathcal{U}=\mathcal{U}_{\rho\rho}+\mathcal{U}_{zz}+\frac1\rho \mathcal{U}_\rho
 =\frac{e^{-u_0(z)}}{\rho^3}+O(\rho^{-1}).
\]
This contradicts $\Delta \mathcal{U}=0$ for $\rho>0$ sufficiently small. We conclude that nondegenerate horizons
cannot occur.

Now let $p_1,\ldots,p_N$ denote the collection of degenerate horizon punctures. Observe that axis analyticity implies that $\mathcal{U}$ extends harmonically across the
axis away from the punctures.  Therefore
\begin{equation} 
 \mathcal{U}>0,
 \qquad
 \Delta \mathcal{U}=0
 \quad\text{on }\R^3\setminus\{p_1,\ldots,p_N\},
 \qquad
 \mathcal{U}\longrightarrow1\quad\text{at infinity}.
\end{equation}
By B\^ocher's theorem \cite{AxlerBourdonRamey}, near each puncture
\begin{equation}
 \mathcal{U}(x)=\frac{m_i}{|\mathbf{x}-p_i|}+h_i(\mathbf{x}),
\end{equation}
where $m_i\geq0$ and $h_i$ is harmonic across $p_i$.  The uniform puncture expansion \eqref{eq:tangentUConvergence} gives
$\mathcal{U}=e^{-U}\to+\infty$ as $\mathbf{x}\to p_i$; hence $m_i>0$.
Moreover, the function
\[
 h(\mathbf x)=\mathcal{U}(\mathbf x)-\sum_{i=1}^N\frac{m_i}{|\mathbf x-p_i|}
\]
extends to a harmonic function on all of $\R^3$, and tends to $1$ at infinity.
It is bounded, so Liouville's theorem implies $h\equiv1$.
\end{proof}

An immediate consequence of the previous two propositions is the following classification.

\begin{corollary}\label{cor:alphaZeroMP}
Let $\Phi=(u,v,\chi,\psi)$ be a stationary electrovacuum singular
harmonic map with black hole rod data in the sense of
Definition~\ref{def:blackHoleData}, and let $\alpha$ be reconstructed
with the asymptotically flat normalization.  If $\alpha\equiv0$, then
all horizon components are degenerate punctures, and the associated spacetime metric
takes the Majumdar-Papapetrou form \eqref{mp000}, with $\mathcal{U}_{mp}=e^{-U}$.
Moreover, after a Heisenberg translation, there is a constant
$\gamma\in[0,2\pi)$ such that
\[
 v=0,\qquad
 d\chi=\cos\gamma\,\rho\starTwo d\mathcal{U}_{mp},\qquad
 d\psi=\sin\gamma\,\rho\starTwo d\mathcal{U}_{mp}.
\]
In particular, the reconstructed Maxwell field is that of a constant
duality rotation of the collinear Majumdar--Papapetrou solution.
\end{corollary}

\subsection{Proof of Theorem~\ref{thm:main}}
Let $M=\max_{1\leq j\leq N+1}\mathbf{b}_j$.  By \eqref{eq:outerDefects}, $M\geq0$.
Suppose first that no bounded axis rod realizes the logarithmic angle defect maximum, $M$.  If $M>0$, then the maximum
could not occur on either semi-infinite rod, whose values are zero, and hence
would have to occur on a bounded rod, leading to a contradiction.  Therefore
$M=0$.  Since no bounded rod realizes $M$, every bounded rod satisfies
$\mathbf{b}_j<0$.  This is alternative~\ref{alt:attractive}.

Suppose instead that some bounded axis rod realizes $M$. By Proposition
\ref{prop:maxRodRigidity} it follows that $\alpha\equiv0$.  We may then appeal to
Corollary \ref{cor:alphaZeroMP}, to obtain the conclusion of alternative~\ref{alt:MP}.


\begin{thebibliography}{99}

\bibitem{AmselEtAl2010}
A.~J. Amsel, G.~T. Horowitz, D. Marolf, and M.~M. Roberts,
\emph{Uniqueness of extremal Kerr and Kerr--Newman black holes},
Phys. Rev. D \textbf{81} (2010), 024033.

\bibitem{AxlerBourdonRamey}
S.~Axler, P.~Bourdon, and W.~Ramey,
B\^ocher's theorem,
Amer. Math. Monthly \textbf{99} (1992), no.~1, 51--55.

\bibitem{BachWeyl1922}
R. Bach and H. Weyl,
\emph{Neue L\"osungen der Einsteinschen Gravitationsgleichungen},
Math. Z. \textbf{13} (1922), 134--145.

\bibitem{BuntingMasood1987}
G.~L. Bunting and A.~K.~M. Masood-ul-Alam,
\emph{Nonexistence of multiple black holes in asymptotically Euclidean static
vacuum space-time},
Gen. Relativity Gravitation \textbf{19} (1987), 147--154.

\bibitem{Carter1971}
B. Carter,
\emph{Axisymmetric black hole has only two degrees of freedom},
Phys. Rev. Lett. \textbf{26} (1971), 331--333.

\bibitem{Carter1973}
B. Carter,
\emph{Black hole equilibrium states},
in \emph{Black Holes/Les Astres Occlus}, Les Houches 1972,
Gordon and Breach, New York, 1973, pp. 57--214.

\bibitem{Carter1985}
B. Carter,
\emph{Bunting identity and Mazur identity for non-linear elliptic systems
including the black hole equilibrium problem},
Comm. Math. Phys. \textbf{99} (1985), 563--591.

\bibitem{ChruscielRigidity}
P.~T. Chru\'sciel,
On rigidity of analytic black holes,
Comm. Math. Phys. \textbf{189} (1997), 1--7.

\bibitem{ChruscielCostaHeusler2012}
P.~T. Chru\'sciel, J.~L. Costa, and M. Heusler,
\emph{Stationary black holes: uniqueness and beyond},
Living Rev. Relativ. \textbf{15} (2012), Art. 7.

\bibitem{ChruscielNadirashvili1995}
P.~T. Chru\'sciel and N.~S. Nadirashvili,
\emph{All electrovacuum Majumdar--Papapetrou space-times with nonsingular
black holes},
Class. Quantum Grav. \textbf{12} (1995), L17--L23.

\bibitem{ChruscielNguyen2010}
P.~T. Chru\'sciel and L. Nguyen,
\emph{A uniqueness theorem for degenerate Kerr--Newman black holes},
Ann. Henri Poincar\'e \textbf{11} (2010), 585--609.

\bibitem{ChruscielReallTod2006}
P.~T. Chru\'sciel, H.~S. Reall, and P. Tod,
\emph{On Israel--Wilson--Perj\'es black holes},
Class. Quantum Grav. \textbf{23} (2006), 2519--2540.

\bibitem{ChruscielReallTodStatic2006}
P.~T. Chru\'sciel, H.~S. Reall, and P. Tod,
\emph{On non-existence of static vacuum black holes with degenerate components
of the event horizon},
Class. Quantum Grav. \textbf{23} (2006), 549--554.

\bibitem{ChruscielTod2007}
P.~T. Chru\'sciel and P. Tod,
\emph{The classification of static electro-vacuum space-times containing an
asymptotically flat spacelike hypersurface with compact interior},
Comm. Math. Phys. \textbf{271} (2007), 577--589.

\bibitem{Costa2010}
J.~L. Costa,
\emph{On the classification of stationary electro-vacuum black holes},
Class. Quantum Grav. \textbf{27} (2010), 035010.

\bibitem{Ernst1968}
F.~J. Ernst,
\emph{New formulation of the axially symmetric gravitational field problem.
II},
Phys. Rev. \textbf{168} (1968), 1415--1417.

\bibitem{HKWXasymptotics}
Q. Han, M. Khuri, G. Weinstein, and J. Xiong,
\emph{Asymptotic analysis of harmonic maps with prescribed singularities},
Ann. PDE \textbf{12} (2026), no.~2, Paper No. 21.

\bibitem{HKWXvacuum}
Q. Han, M. Khuri, G. Weinstein, and J. Xiong,
\emph{Kerr black hole uniqueness}, arXiv:2608.23325 (2026).

\bibitem{HKWXprog} Q. Han, M. Khuri, G. Weinstein, and J. Xiong, in progress.

\bibitem{HartleHawking1972}
J.~B. Hartle and S.~W. Hawking,
\emph{Solutions of the Einstein--Maxwell equations with many black holes},
Comm. Math. Phys. \textbf{26} (1972), 87--101.

\bibitem{IsraelKhan1964}
W. Israel and K.~A. Khan,
\emph{Collinear particles and Bondi dipoles in general relativity},
Nuovo Cimento \textbf{33} (1964), 331--344.

\bibitem{Hennig2019}
J. Hennig,
\emph{On the balance problem for two rotating and charged black holes},
Class. Quantum Grav. \textbf{36} (2019), 235001.

\bibitem{Hennig2020}
J. Hennig,
\emph{Axis potentials for stationary $n$-black hole configurations},
Class. Quantum Grav. \textbf{37} (2020), 19LT01.

\bibitem{Hennig2026}
J. Hennig,
\emph{The balance problem for $n$ aligned black holes},
J. Phys. Conf. Ser. \textbf{3177} (2026), 012022;
arXiv:2604.12134.

\bibitem{HirschZhang} S. Hirsch and Y. Zhang, \textit{Classification of maximally charged black holes}, arXiv:2608.24843 (2026).

\bibitem{KhuriWeinstein2016}
M. Khuri and G. Weinstein,
\emph{The positive mass theorem for multiple rotating charged black holes},
Calc. Var. Partial Differential Equations \textbf{55} (2016), no.~2,
Paper No. 33, 29 pp.

\bibitem{LiTian1992}
Y.~Y. Li and G. Tian,
\emph{Regularity of harmonic maps with prescribed singularities},
Comm. Math. Phys. \textbf{149} (1992), no.~1, 1--30.

\bibitem{LiTian1993}
Y.~Y. Li and G. Tian,
\emph{Harmonic maps with prescribed singularities},
in \emph{Differential Geometry: Partial Differential Equations on Manifolds},
Proc. Sympos. Pure Math., vol. 54, Amer. Math. Soc., Providence, RI, 1993,
pp. 317--326.

\bibitem{Lucietti}
J.~Lucietti,
All higher-dimensional Majumdar--Papapetrou black holes,
Ann. Henri Poincar\'e \textbf{22} (2021), 2437--2450.

\bibitem{Majumdar1947}
S.~D. Majumdar,
\emph{A class of exact solutions of Einstein's field equations},
Phys. Rev. \textbf{72} (1947), 390--398.

\bibitem{Mazur1982}
P.~O. Mazur,
\emph{Proof of uniqueness of the Kerr--Newman black hole solution},
J. Phys. A \textbf{15} (1982), 3173--3180.

\bibitem{Meinel2012}
R. Meinel,
\emph{Constructive proof of the Kerr--Newman black hole uniqueness including
the extreme case},
Class. Quantum Grav. \textbf{29} (2012), 035004.


\bibitem{Nguyen2011}
L. Nguyen,
\emph{Singular harmonic maps and applications to general relativity},
Comm. Math. Phys. \textbf{301} (2011), no.~2, 411--441.

\bibitem{Papapetrou1947}
A. Papapetrou,
\emph{A static solution of the equations of the gravitational field for an
arbitrary charge distribution},
Proc. Roy. Irish Acad. Sect. A \textbf{51} (1947), 191--204.

\bibitem{Robinson1975}
D.~C. Robinson,
\emph{Uniqueness of the Kerr black hole},
Phys. Rev. Lett. \textbf{34} (1975), 905--906.

\bibitem{SudarskyWald}
D.~Sudarsky and R.~M.~Wald,
\emph{Mass formulas for stationary Einstein--Yang--Mills black holes and a simple proof of two staticity theorems},
Phys.\ Rev.\ D \textbf{47} (1993), R5209--R5213.


\bibitem{Weinstein1990}
G. Weinstein,
\emph{On rotating black holes in equilibrium in general relativity},
Comm. Pure Appl. Math. \textbf{43} (1990), no.~7, 903--948.

\bibitem{Weinstein1992}
G. Weinstein,
\emph{The stationary axisymmetric two-body problem in general relativity},
Comm. Pure Appl. Math. \textbf{45} (1992), no.~9, 1183--1203.


\bibitem{Weinstein1995}
G. Weinstein,
\emph{On the Dirichlet problem for harmonic maps with prescribed
singularities},
Duke Math. J. \textbf{77} (1995), no.~1, 135--165.

\bibitem{Weinstein1996}
G. Weinstein,
\emph{$N$-black hole stationary and axially symmetric solutions of the
Einstein/Maxwell equations},
Comm. Partial Differential Equations \textbf{21} (1996), no.~9--10,
1389--1430.

\bibitem{Weinstein} G. Weinstein, \emph{Harmonic maps with prescribed singularities and applications in general relativity}, in {\it The role of metrics in the theory of partial differential equations}, Adv. Stud. Pure Math., \textbf{85} (2020), 479--489.

\end{thebibliography}
\end{document}